\documentclass[twocolumn]{autart} 

\usepackage[utf8]{inputenc}
\usepackage[english]{babel}
\usepackage{amsmath}
\usepackage{amsfonts}
\usepackage{amssymb}
\usepackage{mathtools}
\usepackage{subfig}
\usepackage[table]{xcolor}
\usepackage[comma, numbers]{natbib}
\usepackage{graphicx}  
\newtheorem{theorem}{Theorem}

\newtheorem{definition}[theorem]{Definition}
\newtheorem{example}[theorem]{Example}

\newtheorem{remark}[theorem]{Remark}

\newtheorem{corollary}[theorem]{Corollary}
\numberwithin{theorem}{section}
\newenvironment{proofof}{{\em Proof of }}{\hfill \hspace*{1pt}
\hfill $\blacksquare$}
\newenvironment{proof}{{\em Proof.}}{\hfill \hspace*{1pt}
\hfill $\blacksquare$}
\newcommand\mymatrix[2]{\left[\begin{array}{#1} #2 \end{array}\right]}
\newcommand\RE{\mathbb{R}}

\DeclareMathOperator{\boundary}{bd}

\DeclareMathOperator{\sector}{Sector}

\begin{document}
\begin{frontmatter}
\title{Analysis of Lur'e dominant systems in the frequency domain} 
\thanks[footnoteinfo]{The research leading to these results has received funding from the European Research Council under the
Advanced ERC Grant Agreement Switchlet n.670645.}
\author[Cambridge]{F. A. Miranda-Villatoro}\ead{fam48@cam.ac.uk},   
\author[Cambridge]{F. Forni}\ead{f.forni@eng.cam.ac.uk},             
\author[Cambridge]{R. Sepulchre}\ead{f.forni@eng.cam.ac.uk}  
\address[Cambridge]{University of Cambridge, Department of Engineering, Trumpington Street, Cambridge CB2 1PZ}                                     
\begin{abstract}                          
Frequency domain analysis of linear time-invariant (LTI) systems in feedback with
static nonlinearities is  a classical and fruitful topic of nonlinear systems theory.
We generalize this approach beyond equilibrium stability analysis with the aim of characterizing feedback 
systems whose asymptotic behavior is low dimensional. We illustrate the theory with a generalization of the 
circle criterion for the analysis of multistable and oscillatory Lur'e feedback systems.
\end{abstract}
\end{frontmatter}
\section{Introduction} \label{sec:intro}
Feedback systems that admit the representation in Figure \ref{fig:lureSys} became seminal 
since the formulation of the absolute stability problem
by Lur'e and Postnikov \cite{lure1944}. They have stimulated a great deal of research in 
nonlinear stability theory. Milestones include (i) the formulation of graphical conditions
in the complex plane to determine necessary and sufficient conditions for absolute stability 
(e.g. the Nyquist criterion and the circle criterion), see e.g. \cite{zames1966a, zames1966b};
(ii) the equivalence between those
frequency-domain conditions and linear matrix inequalities (the KYP lemma) \cite{brogliato2007}; and (iii)
the characterization of input-output properties of nonlinear systems through dissipation
inequalities (dissipativity theory  \cite{willems1972} and Integral Quadratic Constraints 
\cite{megretski1997}). 
Collectively, those developments have provided a rich analytical and computational framework to assist
the search of a quadratic Lyapunov function in the global stability analysis of the feedback system.
\begin{figure}[hbt]
	\centering
	\includegraphics[width=0.45\columnwidth]{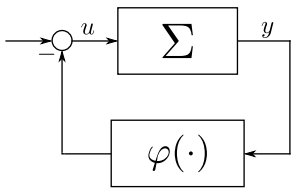}
	\caption{Lur'e feedback system}
	\label{fig:lureSys}
\end{figure}
The present paper seeks to mimic this analysis in Lur'e feedback systems that possess
more general attractors than a single equilibrium. In particular, we are motivated by the analysis
of multistable systems (systems with several stable equilibria) and oscillatory systems 
(systems with a stable limit cycle). 
Our analysis is differential: we look for a
quadratic storage function that decreases along the solutions of the {\it linearized} system along
arbitrary trajectories, possibly rescaled in time.  When such a  storage is positive definite, it is a {\it differential} Lyapunov function 
and serves to prove global contraction to a unique equilibrium \cite{forni2013}. 
Here we relax the positivity condition of the storage. Instead,
we impose that the quadratic form has a given inertia : $p$ negative eigenvalues and
$n-p$ positive eigenvalues. We use this storage to prove absolute $p$-dominance of the Lur'e system, 
that is, contraction to a $p$-dimensional attractor for all nonlinearities satisfying a differential sector condition. 
When $p = 1$, the analysis is relevant to study multistability (the
one dimensional attractor being a heteroclinic orbit). When $p = 2$, the analysis is relevant
to study limit cycles. For $p = 0$, the analysis boils down to the differential analysis of the
absolute stability problem, pioneered by Kalman \cite{kalman1957} and studied by many authors 
since then \cite{jouffroy2010, lohmiller1998, pavlov2005}.
The main message of our paper is that the many classical tools developed in absolute stability theory 
to construct quadratic Lyapunov functions  carry over the construction of quadratic storage with a given inertia.
This means that the analysis of absolute p-dominance closely resemble the analysis of absolute
stability, simply replacing the requirement of positive definiteness of the storage by a requirement 
on the inertia. 
The paper is organized as follows: the next section recalls the results of \cite{forni2017}
dealing with $p$-dominance in the time domain. Afterwards, the analysis of $p$-dominance in the 
frequency domain is studied in Section \ref{section:pdominance:frequency}, with
the extensions of the Nyquist criterion and the 
Kalman-Yakubovich-Popov lemma for the study of $p$-passivity as main results. 
These results are applied to the study of Lur'e feedback systems
in the differential framework in Section \ref{section:lure:frequency}, together with the 
generalization of the circle criterion. Section \ref{section:lure:differential} illustrates the
application of the differential framework to the study of low dimensional attractors through
some examples. Finally, we end the paper with a brief section presenting conclusions and future research.  
{
\small
\textbf{Notation.}
Let $A \in \RE^{n \times n}$ be a real matrix, we denote the spectrum of $A$ by
$\sigma(A)$ which is a subset of the set of complex numbers denoted by $\mathbb{C}$.
The sets $\mathbb{C}_{+} := \left\{ s \in \mathbb{C} | \Re \{s \} > 0 \right\}$,
$\mathbb{C}_{-} := \left\{ s \in \mathbb{C} | \Re\{s \} < 0 \right\}$ denote the 
half spaces in the complex plane with positive and negative real parts respectively. 
The inertia of a matrix $A \in \RE^{n \times n}$ is a set of three elements containing the 
number of eigenvalues with positive, zero and negative real parts, respectively, that is,
a matrix with inertia $(p, 0, n-p )$ has $p$ eigenvalues in $\mathbb{C}_{+}$, $0$ eigenvalues
on the $j \omega$-axis and $n - p$ eigenvalues in $\mathbb{C}_{-}$.  
Let $s \in \mathbb{C}$ be a complex number, then $\overline{s}$ denotes the complex conjugate
of $s$. In the same way, if $z \in \mathbb{C}^{n}$, then $z^{*}$ stands for the complex conjugate
transpose operator, that is, $z^{*} = \overline{z}^{\top}$, where the complex conjugate of a
vector is taken element-wise.
}
\section{Dominance in the time domain}
\label{section:pdominance:time}
\subsection{Dominance of linear systems}
A  linear system
\begin{equation}
\label{eq:linear}
\dot{x} = Ax \quad\qquad x\in\RE^n
\end{equation}
is $p$-dominant if its behavior is characterized by $p$ dominant modes and $n-p$ transient modes, that is, if its asymptotic behavior is $p$-dimensional. 
This property is of special interest if $p$ is much smaller than $n$.  One characterization of  the property is via  a matrix inequality with inertia constraints.
\begin{definition}
\label{def:p-dominance}
The linear system \eqref{eq:linear} is \emph{$p$-dominant 
with rate $\lambda \geq 0$} if and only if there exist 
a symmetric matrix $P$ with inertia $(p,0,n-p)$ and $\varepsilon \geq 0$ such that
\begin{equation}
\label{eq:p-dominance}
\mymatrix{c}{\dot{x} \\ x}^{\top}\!
\mymatrix{cc}{
0 & P \\ P &  \ 2\lambda P + \varepsilon I
}
\mymatrix{c}{\dot{x} \\ x}
\leq 0
\end{equation}
for all $x \in \RE^n$.  The property is {\it strict} if $\varepsilon >0$.
\end{definition}
The dissipation inequality \eqref{eq:p-dominance} is equivalent to the linear matrix
inequality \begin{equation}
\label{eq:LMI-dominance}
A^{\top} P + P A + 2 \lambda P \leq -\epsilon I \, ;
\end{equation} 
which corresponds to a standard Lyapunov inequality
for the shifted matrix $A+\lambda I$, with $P$ of a given inertia. 
Following \cite[Theorem 1]{forni2017}, for $\varepsilon > 0$ 
the constraint $(p,0,n-p)$ on the inertia of $P$ guarantees
that the matrix $A+ \lambda I $ has $p$ eigenvalues 
with strictly positive real part and $n-p$ eigenvalues with strictly negative real part.
Indeed, the original linear system $\dot{x} = Ax$ has $n-p$ transient modes
whose convergence to zero is bounded by the rate $\lambda$.
The remaining $p$ modes may converge to zero at exponential rate or even diverge. 
Those modes dominate the system behavior.
Dissipativity theory \cite{willems1972}  provides an extension of the conic constraint \eqref{eq:p-dominance}
to open systems with the state-space representation 
\begin{equation}
\label{eq:linear-open}
\begin{cases}
	\dot{x} = A x + B u  \\
	y = C x + D u	
\end{cases}
\end{equation}
where $x\in \RE^n$ and $w := (y,u) \in \RE^{m_1} \times \RE^{m_2}$.
\begin{definition}
\label{def:p-dissipativity}
The linear system \eqref{eq:linear-open}
is \emph{$p$-dissipative with rate $\lambda \geq 0$ 
and supply}
\begin{equation}
\label{eq:supply}
s(w) := \mymatrix{c}{y \\ u}^{\!\top}  \!\!
\mymatrix{cc}{
Q & L \\ L^{\top} & R
} \!
\mymatrix{c}{y \\ u}
\end{equation} 
if there exists a symmetric matrix $P$ with inertia $(p,0,n-p)$
and $\varepsilon \geq 0$ such that 
\begin{equation}
\label{eq:p-dissipativity}
\mymatrix{c}{\dot{x} \\  x}^{\!\top}\!\!
\mymatrix{cc}{
0 & P \\ P & 2\lambda P + \varepsilon I
}\!
\mymatrix{c}{\dot{x} \\ x}\!
\leq\!
\mymatrix{c}{y \\ u}^{\!\top} \!\!
\mymatrix{cc}{
Q & L \\ L^{\top} & R
} \!
\mymatrix{c}{y \\ u}
\end{equation}
for all  $x \in \RE^n$ and all $w \in \RE^{m_1+m_2}$.
The property is strict if $\varepsilon > 0$.
\end{definition}
\eqref{eq:p-dissipativity} is a standard dissipation inequality 
for the shifted system $\dot{x} = (A+\lambda I) x + B u$, $y=Cx + Du$.
The classical interpretation \cite{willems1972b}
is that the supply rate bounds the variation of the storage $x^{\top} P x$ along the
trajectories of the shifted system. If the storage is positive definite, dissipativity 
implies stability if the supply is nonpositive. This internal property is replaced by
$p$-dominance in the case of $p$-dissipativity. For instance,
following \cite{forni2017} and \cite{Forni2017b}, we observe that 
\eqref{eq:p-dissipativity} is equivalent to the feasibility of the 
linear matrix inequality
{\footnotesize
\begin{equation}
\label{eq:lmi:pPassive}
\mymatrix{cc}
{
\!\!A^{\top} \! P \!+\! P\! A  \!-\!C^{\top}  \! Q C \!+ \!2\lambda P \!+\! \varepsilon I \!\! &\!\! \! P B \!-\! C^{\top}\!L-C^{\top}  \! Q D \!\!\! \\
 \!\! B^{\top} \!P \!-\! L^{\top} C-D^{\top}  \! Q C \!\!  \!\!& \!\!  -\!R\!-\!D^{\!\top}\!\! L \!-\!L^{\!\top}\!\! D \!-\! D^{\top}\!QD\!\!\!
 }
\!\!\leq\! 0.
\end{equation}
}
A necessary condition is that the top-left element in \eqref{eq:lmi:pPassive} is non positive which,
for $Q=0$, guarantees $p$-dominance of \eqref{eq:linear}.
The supply rate \eqref{eq:supply} embraces many cases studied in the classical literature of dissipative systems. 
For example, the case 
$Q = R = 0$ and $L = I$ identifies the class of \emph{$p$-passive} systems.
Likewise, $Q = - \varepsilon I$ and 
$R = - \delta I$ with $L = I$ correspond to 
\emph{strictly output $p$-passive}
and \emph{strictly input $p$-passive} systems, respectively.
\subsection{Dominance of nonlinear systems}
Dominance is defined for nonlinear systems by making the analysis  {\it differential}
\cite{forni2014,forni2017,lohmiller1998}, i.e. by requiring dominance of 
the linearized system equations along arbitrary trajectories.
For the nonlinear system 
\begin{equation}
\label{eq:system}
\dot{x} = f(x)  \qquad \quad x\in \RE^n
\end{equation}
the prolonged system \cite{Crouch1987}
\begin{equation}
\label{eq:prolonged}
\left\{
\begin{array}{rcl}
\dot{x} &=& f(x) \\
\dot{\delta x} &=& \partial f(x) \delta x 
\end{array}
\right.
\qquad (x,\delta x)\in \RE^n \!\times\! \RE^n
\end{equation}
defines the linearization of the system along any
of its trajectories. The intuition is that the variational trajectories $\delta x(\cdot)$ 
capture the nonlinear behavior in an infinitesimal neighborhood of any trajectory $x(\cdot)$,
\cite{forni2014,Forni2014b}.
\begin{definition}
\label{def:diff-p-dominance}
A nonlinear system \eqref{eq:system} is \emph{$p$-dominant 
with rate $\lambda\geq0$} and constant symmetric matrix $P$ with inertia $(p,0,n-p)$
if there exists a  constant $\varepsilon \geq 0$ for which the prolonged system \eqref{eq:prolonged}
satisfies the conic constraint
\begin{equation}
\label{eq:diff-p-dominance}
\mymatrix{c}{\!\dot {\delta x}\!\! \\ \!\delta x\!\!}^{\top} \!
\mymatrix{cc}{
0 & P \\ P & \ 2\lambda P + \varepsilon I
}
\mymatrix{c}{\!\dot {\delta x}\!\! \\ \!\delta x\!\!}
\leq 0
\end{equation}
for all $(x,\delta x)\in \RE^n \times \RE^n$ The property is strict if $\varepsilon>0$.
\end{definition}
Indeed, \eqref{eq:diff-p-dominance} is equivalent to finding a uniform solution $P$ 
to the linear matrix inequality 
\begin{equation}
\label{eq:LMI-dominance:diff}
\partial f(x)^{\top} P + P \partial f(x) + 2 \lambda P \leq -\epsilon I 
\end{equation} 
for each $x\in \RE^n$.
Dominance strongly restricts the asymptotic behavior of the nonlinear system.
$0$-dominant systems are contractive \cite{lohmiller1998, pavlov2005,Forni2014b},
thus all bounded trajectories converge to a unique fixed point. 
$1$-dominant systems are 
monotone systems \cite{Hirsch2006,Angeli2003}, their attractors must
preserve a suitable order relation, which enforces generic convergence to fixed points.
$2$-dominant systems make contact with the property of monotonicity
with respect to rank-2 cones, exploited in \cite{Smith1980,Sanchez2009} to characterize
periodic attractors. The reader is referred to \cite{Forni2017b} for a thorough comparison
with the literature and for a detailed analysis of the behavior of a $p$-dominant
systems. The following results from \cite{Forni2017b} justifies the interest of this paper for 
$p$-dominance.
\begin{theorem}
\label{thm:reduced}
Let  \eqref{eq:system} be a 
strictly $p$-dominant system with rate $\lambda \geq 0$.
Then, the flow on any compact $\omega$-limit set is topologically equivalent 
to a flow on a compact invariant set of a Lipschitz system in $\RE^p$.
\end{theorem}
The theorem captures the property that the asymptotic behavior of a $p$-dominant system
is $p$-dimensional. For $p \leq 2$, $p$-dominance strongly constrains the possible attractors.
\begin{corollary}
\label{thm:asymptotic_behavior}
Under the assumptions of Theorem \ref{thm:reduced},
every bounded solution of \eqref{eq:system} asymptotically converge to
\begin{itemize}
\item a unique fixed point if $p = 0$;
\item a fixed point if $p = 1$;
\item a simple attractor if $p=2$, i.e. a fixed point, a set of fixed points and their connected arcs, or a limit cycle. 
\end{itemize}
\end{corollary}
In what follows we will study $p$-dominant systems arising from the interconnection 
of $p$-dissipative linear systems with static nonlinearities satisfying a differential sector condition.
Theorem \ref{thm:reduced} and Corollary \ref{thm:asymptotic_behavior} will be particularly useful
to predict the asymptotic behavior of those closed-loop systems. 
\section{Dominance in the frequency domain}
\label{section:pdominance:frequency}
\subsection{Nyquist criterion for $p$-dominance}
Under the standing assumption of minimal realization, 
the poles of the transfer function $G(s) := C(sI-A)^{-1} B +D$
of a strict $p$-dominant system with non-negative rate $\lambda$ 
are separated in two groups. The $p$ dominant poles belong to the interior of the Nyquist region 
\begin{equation}
\label{eq:nyquist_region}
	\Omega_{\lambda} := \left\{ s \in \mathbb{C} | s+ \lambda \in \mathbb{C}_{+} \right\}.
\end{equation}
while the  remaining $n-p$ poles belong to its complement. This follows directly from
the separation of eigenvalues of the matrix $A+\lambda I$ into unstable and stable groups,
which guarantees that the shifted transfer function 
$G(s-\lambda) = C(sI-A-\lambda I)^{-1} B +D$ has $p$-poles in $\mathbb{C}_{+}$ and
$n-p$ poles in $\mathbb{C}_{-}$.
The Nyquist criterion is a cornerstone of control theory for the study of closed loop stability.
The principle of the argument relates the Nyquist 
locus of the open loop system to the position of the poles in closed loop, 
providing graphical conditions for closed loop stability. A similar approach can be 
pursued for $p$-dominance. A straightforward application
of the principle of the argument adapted to the Nyquist region $\Omega_{\lambda}$
provides conditions for closed loop $p$-dominance based on the
Nyquist locus of the open loop transfer function. 
\begin{theorem}[Nyquist dominance criterion] \label{theorem:nyquistCriterion}
	Let $G(s)$ be the (SISO) transfer function of a strict $p_1$-dominant system with rate $\lambda\geq 0$.
	Then, the closed-loop system 
	$G(s)/(1 + kG(s))$ is strictly $p_2$-dominant with rate $\lambda$ if and only if
	the Nyquist plot of $G(s)$ computed along the boundary of $\Omega_{\lambda}$, encircles the point
	$-1/k$, $(p_2-p_1)$-times in the clockwise direction.
\end{theorem}
\begin{proofof}\emph{Theorem \ref{theorem:nyquistCriterion}.}
	The proof follows  from Cauchy's principle of the argument \cite[Section 6.2]{marsden1999},
	by taking $f(s) = 1 + k G(s)$ and counting the encirclements $E$ of the origin as $s$ moves along 
	the boundary $\boundary \Omega_{\lambda}$ of $\Omega_{\lambda}$.
	Equivalently, we observe that 
	the Nyquist plot of $G(s)$ along the Nyquist path $\boundary \Omega_{\lambda}$ is the same as the 
	Nyquist plot of $G(s-\lambda)$ along the Nyquist path $s=j\omega$, $\omega \in \RE$. 
	Thus, the encirclements can be counted as $s - \lambda$ varies along the $j \omega$-axis.
	
	Define the closed loop transfer function $W(s) := \frac{G(s)}{1 + k G(s)}$ given by a negative
	feedback loop around $G(s)$ with gain $k>0$. Take $$f(s) := 1 + k G(s)$$ and 
	consider $G(s)$ as the ratio of two polynomials $n(s)/d(s)$. Then, 
        $G(s-\lambda) = n(s - \lambda) / d(s - \lambda)$ and 
	\begin{displaymath}
		f(s-\lambda) = 1 + k G(s-\lambda) = \frac{d(s - \lambda) + k n(s - \lambda)}{d(s - \lambda)}.
	\end{displaymath}
	Indeed the poles of $f(s-\lambda)$ correspond to the poles of $G(s-\lambda)$.
	The zeros of $f(s-\lambda)$ correspond to the shifted closed-loop poles,
	given by $W(s-\lambda)$. Hence, the number of clockwise encirclements 
	$E$ is  the difference between the number of zeros and poles of $f(s-\lambda)$ in $\mathbb{C}_+$,
	that is, $E = p_2 - p_1$. The result follows.
\end{proofof}

Theorem \ref{theorem:nyquistCriterion} shows that dominance of the closed loop can be determined from the 
Nyquist locus of the shifted transfer function $G(s-\lambda)$. The degree of closed loop dominance $p_2=E-p_1$
is given by the difference between the clockwise encirclements of the locus around $-\frac{1}{k}$ 
minus the number of unstable poles of $G(s-\lambda)$.
The classical Nyquist criterion for stability is recovered from the previous
theorem by taking $\lambda = 0$ and $p_2 = 0$. 
We note that the theorem still applies when $G(s-\lambda)$ has either poles or zeros on the
$j \omega$-axis,  using the standard indentation technique
along the boundary of $\Omega_{\lambda}$.
\begin{example}
	Consider the linear system 
	\begin{equation}
		\label{eq:nyquist:sys}
		G(s) = \frac{10}{(s^2 + 2 s + 2)(s + 3)}.
	\end{equation}
	The poles of $G(s)$ are $-1 \pm j$ and $-3$. For $\lambda = 2.5$,
	$G(s-\lambda)$ has $2$ poles in $\mathbb{C}_{+}$ and 1 pole in $\mathbb{C}_{-}$.
	The Nyquist plot of $G(s)$ along the boundary of $\Omega_{\lambda}$ is in 
	Figure \ref{fig:nyquist_ex}.  For any positive $k$, there are no encirclements of the point 
	$-\frac{1}{k}$. Thus, the closed loop system formed by the negative feedback of \eqref{eq:nyquist:sys} 
	with any static gain $k>0$ is $2$-dominant with rate $\lambda = 2.5$.
		\begin{figure}[htpb]
			\centering
			\includegraphics[width=0.4\textwidth]{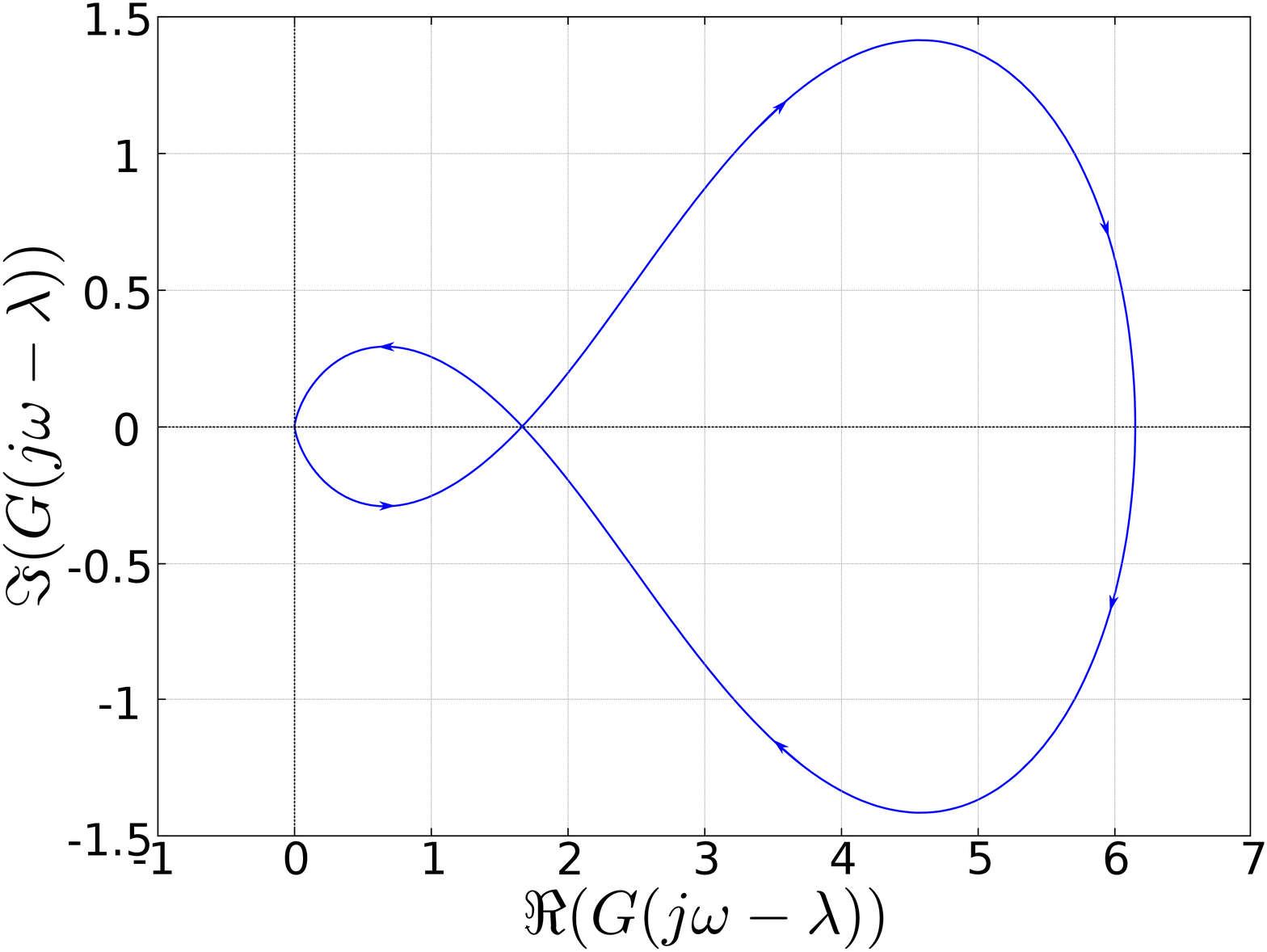}
			\caption{Modified Nyquist diagram with $\lambda = 2.5$ of the system 
				\eqref{eq:nyquist:sys}.}
			\label{fig:nyquist_ex}
		\end{figure}
\end{example}
The Nyquist criterion for dominance provides graphical conditions for dominance of the feedback system.
 More fundamentally, Theorem \ref{theorem:nyquistCriterion} 
shows that dominance, like stability,  is shaped by feedback. 
Compensators can be introduced to shape the open loop transfer function 
with the goal of modulating the degree of dominance of the closed loop system. Likewise, dominance
can be made {\it robust}, by defining dominance margins in the same way as
stability margins. 
\subsection{Kalman-Yakubovich-Popov lemma} 
A key feature of dissipativity is the equivalence of the property in the time domain and in the frequency domain,
characterized by the celebrated Kalman-Yakubovich-Popov lemma \cite{rantzer1996}. 
The same equivalence holds for $p$-dissipativity.
\begin{theorem}[KYP lemma for $p$-dissipativity]   $ $	
	A linear system is $p$-dissipative with rate $\lambda$ and supply \eqref{eq:supply} 
	if and only if, for all $\omega \in \RE \cup \{\infty \}$ with
	$\det(j \omega I - (A + \lambda I)) \neq 0$, its shifted transfer
	function $G(s-\lambda)$ has $p$ poles in $\mathbb{C}_{+}$ and satisfies
	\begin{equation}
		\begin{bmatrix}
			G(j \omega-\lambda) \\ I
		\end{bmatrix}^{*} 
		\begin{bmatrix}
			Q & L 
			\\
			L^{\top} & R
		\end{bmatrix}
		\begin{bmatrix}
			G(j \omega-\lambda) \\ I
		\end{bmatrix} \geq 0.
		\label{eq:pdissipative:freq}
	\end{equation}
	\label{th:kypLemma}
\end{theorem}
\begin{proofof}\emph{Theorem \ref{th:kypLemma}.}
	\underline{Sufficiency}: Assume that the system is $p$-dissipative, and consider the complex
	input $u(t) := \exp(j \omega t)u_{0}$, $u_{0} \in U$. It follows that 
	$x^{\lambda}(t) := \exp(j \omega t) (j \omega I - (A + \lambda I))^{-1} B u_{0}$ 
	is a solution of the shifted system $\dot{x} = A_\lambda x + Bu$, $y= Cx + Du$, where 
	$A_{\lambda} := A + \lambda I$. The following identity must then hold
	(we remove the dependence on $t$ for readability):
	\begin{align*}
		x^{\lambda *} P B u + u^{*} B^{\top} P x^{\lambda}
		& = u_{0}^{\top}
		B^{\top} \left( -j \omega - 
		A_{\lambda}^{\top} \right)^{-1} P B u_{0} 
		\\
	& + u_{0}^{\top} B^{\top} P \left(j \omega I - A_{\lambda} \right)^{-1} B u_{0}
		\\
		& = - x^{\lambda *} \left( A_{\lambda}^{\top} P + P 
		A_{\lambda} \right) x^{\lambda}
	\end{align*}
	From this and \eqref{eq:lmi:pPassive} it follows that
	\begin{align*}
		0 & \geq \begin{bmatrix}
			x^{\lambda} \\ u
		\end{bmatrix}^{*} 
		\left(
		\begin{bmatrix}
			\left( A + \lambda I \right)^{\top} P + P (A + \lambda I) & PB
			\\
			B^{\top} P & 0
		\end{bmatrix} 
		\right.
		\\
		& - \left. 
		\begin{bmatrix}
			C^{\top} & 0
			\\
			D^{\top} & I
		\end{bmatrix}
		\begin{bmatrix}
			Q & L
			\\
			L^{\top} & R
		\end{bmatrix}
		\begin{bmatrix}
			C & D
			\\
			0 & I
		\end{bmatrix}
		\right)
		\begin{bmatrix}
			x^{\lambda} \\ u
		\end{bmatrix} 
		\\
		& = 
		- u_{0}^{\top} \begin{bmatrix}
			G(j \omega-\lambda) \\ I
		\end{bmatrix}^{*}
		\begin{bmatrix}
			Q & L
			\\
			L^{\top} & R
		\end{bmatrix}
		\begin{bmatrix}
			G(j \omega-\lambda) \\ I
		\end{bmatrix} u_{0}.
	\end{align*}
	Recalling that $u_{0} \in U$ is arbitrary, \eqref{eq:pdissipative:freq} follows.
	
	\underline{Necessity}: as in the works \cite{balakrishnan2003, rantzer1996}.
\end{proofof}

If  the $p$-dissipativity is strict, then \eqref{eq:pdissipative:freq} is also
a strict inequality \cite{balakrishnan2003}.
For $p = 0$ Theorem \ref{th:kypLemma}
reduces to the standard KYP lemma, where the matrix $P$ is positive semidefinite.
A classical result from the theory of passive systems, $Q = R = 0$ and $L = I$, 
states that the transfer function of a passive system is positive real. 
For $p$-passivity we still have that 
$$
G(j \omega-\lambda)^* + G(j \omega-\lambda) \geq 0 \qquad \forall \omega \in \RE \cup \{\infty\} \, .
$$
The reader will notice this positive realness property on  the Nyquist locus in Figure \ref{fig:nyquist_ex}.
\subsection{Geometric conditions for $p$-passivity}
The Kalman-Yakubovich-Popov lemma provides necessary geometric conditions 
for $p$-passive systems in terms of the relative degree of the transfer functions and
of the position of its zeros. For the next theorem we assume a single-input single-output system.
\begin{theorem} \label{theorem:p:geom}
	Let \eqref{eq:linear-open} be a SISO $p$-passive system with rate $\lambda$. Then the relative degree
	$\Delta$
	of its transfer function $G(s)$	is less than or equal to $2p + 1$ and it satisfies
	\begin{equation}
		r \in \left[ p - \frac{\Delta + 1}{2}, p - \frac{\Delta - 1}{2} \right] \cap \mathbb{Z}_{+},
		\label{eq:nqrp:formula}
	\end{equation}
	where,
	\begin{enumerate}
		\item[] $\Delta$ = relative degree of $G(s)$, i.e., $\Delta = n - q$.
		\item[] $n$ = total number of poles of $G(s-\lambda)$.
		\item[] $q$ = total number of finite zeros of $G(s-\lambda)$.
		\item[] $p$ = number of poles of $G(s-\lambda)$ in $\mathbb{C}_{+}$.
		\item[] $r$ = number of finite zeros of $G(s-\lambda)$ in $\mathbb{C}_{+}$.
	\end{enumerate}
\end{theorem}
\begin{proofof}\emph{Theorem \ref{theorem:p:geom}.}
	We first recall that static feedback does not affect neither the relative degree, nor
	the position of the zeros of a linear system. Secondly, static negative output feedback of
	a $p$-passive system preserves the $p$-passivity of the closed-loop. Hence, the 
	closed-loop transfer function 
	\begin{displaymath}
		H(s) = \frac{G(s)}{1 + K G(s)} = \kappa \frac{n(s)}{d_{1}(s) d_{2}(s)},
	\end{displaymath}
	is $p$-passive whenever $G(s)$ is $p$-passive for any $K \geq 0$, where $n(s)$, $d_{1}(s)$ and
	$d_{2}(s)$ are polynomials of degree $q$, $q$ and $n-q$ respectively.
	Moreover, as $K$ increases, $q$ poles of $H(s)$ become closer to its $q$ zeros, i.e., 
	$\frac{n(s)}{d_{1}(s)} \to 1$ as $K \to +\infty$ . This last fact
	is commonly used in the construction of the root-locus \cite[Chapter 6]{engelberg2005}.
	Assuming that $K$ is sufficiently large, the phase of $H(s)$ \emph{approximates} to
	\begin{displaymath}
		\arg(H(j \omega - \lambda)) = \arg(\kappa) - \arg(d_{2}(j \omega - \lambda)).
	\end{displaymath}
	Because $r$ poles of $H(s)$ move towards $r$ zeros
	in $\Omega_{\lambda}$ and $H(s)$ is $p$-passive, it follows that $p-r$ roots of $d_{2}(s-\lambda)$
	must lie in the interior of $\mathbb{C}_{+}$. As a  consequence
	\begin{displaymath}
		\arg(H(j 0 - \lambda)) = \arg(\kappa) - (p - r) \pi
	\end{displaymath}
	and 
	\begin{displaymath}
		\arg(H(j \omega - \lambda)) \to \arg(\kappa) - \Delta \frac{\pi}{2}, \text{ as } \omega
		\to +\infty.
	\end{displaymath}
	By Theorem \ref{th:kypLemma}, it follows that $p$-passivity of $H(s)$ implies
	$\arg(H(j 0 - \lambda)) = 2 \pi \eta$
	for some $\eta \in \mathbb{Z}$ and $\arg(H(j \omega - \lambda)) \in [-\frac{\pi}{2}, \frac{\pi}{2}] +
	2 \pi \eta$ for all $\omega \in \RE \cup \{\infty\}$. Hence,
	\begin{multline*}
		\arg(H(j \infty-\lambda)) \in \left[ -\frac{\pi}{2}, \frac{\pi}{2} \right]
		+ \arg(H(j 0-\lambda)),
	\end{multline*}
	or equivalently
	\begin{displaymath}
		- \Delta \frac{\pi}{2} \in \left[ -\frac{\pi}{2}, \frac{\pi}{2} \right] - (p - r) \pi.
	\end{displaymath}
	By noticing that $\Delta,r$ and $p$ are all positive integers, the previous condition transforms into
	\eqref{eq:nqrp:formula}. To conclude the proof, assume by contradiction that 
	$\Delta = 2p + \nu$, where $\nu \geq 2$, $\nu \in \mathbb{N}$. It follows that 
	the right-hand side of \eqref{eq:nqrp:formula} becomes the empty set, and
	we get the desired contradiction. Therefore, $\Delta \leq 2p + 1$.
	This ends the proof.
\end{proofof}

Table \ref{table:ppassivity:conditions} shows necessary conditions for $p$-passivity with
rate $\lambda$ in terms of $n, q, r$, for several values of the dominance degree $p$. 
In the case of classical passivity, that is, $0$-passivity, we recover the well-known necessary condition
of passivity: a passive system must be minimum phase ($r = 0$) and its relative degree is at most one ($n - q \leq 1$). 
The poles and zeros in Table \ref{table:ppassivity:conditions} refer to
$G(s-\lambda)$. In consequence their position in the
complex plane depend on the value of $\lambda$. 
\begin{table}[htpb]
\centering
\begin{tabular}{ |c|c|c|c|c| }
\hline
$\mathbf{n - q}$ & $\mathbf{p = 0}$ & $\mathbf{p = 1}$ & $\mathbf{p = 2}$ & $\mathbf{p = 3}$ \\
\hline
\textbf{0} & $r \in \{0\}$ & $r \in \{1\}$ & $r \in \{2\}$ & $r \in \{3\}$ \\
\textbf{1} & $r \in \{0\}$ & $r \in \{0, 1\}$ & $r \in \{1,2\}$ & $r \in \{2, 3\}$ \\
\textbf{2} & \cellcolor[HTML]{CCCCCC} & $r \in \{0\}$ & $r \in \{1\}$ & $r \in \{2\}$ \\
\textbf{3} & \cellcolor[HTML]{CCCCCC} & $r \in \{0\}$ & $r \in \{0, 1\}$ & $r \in \{1,2\}$ \\
\textbf{4} & \cellcolor[HTML]{CCCCCC} & \cellcolor[HTML]{CCCCCC} & $r \in \{0\}$ & $r \in \{1\}$ \\
\textbf{5} & \cellcolor[HTML]{CCCCCC} & \cellcolor[HTML]{CCCCCC} & $r \in \{0\}$ & $r \in \{0,1\}$ \\
\textbf{6} & \cellcolor[HTML]{CCCCCC} & \cellcolor[HTML]{CCCCCC} & \cellcolor[HTML]{CCCCCC} & $r \in \{0\}$ \\
\textbf{7} & \cellcolor[HTML]{CCCCCC} & \cellcolor[HTML]{CCCCCC} & \cellcolor[HTML]{CCCCCC} & $r \in \{0\}$ \\
\hline 
\end{tabular}
\vspace{1mm}
\caption{Necessary conditions for $p$-passivity in terms of the relative degree, and the poles
	and zeros of $G(s-\lambda)$ in $\mathbb{C}_{+}$.}
	\label{table:ppassivity:conditions}
\end{table}
\begin{example} \label{ex:3-poles}
	Consider a system with $3$ real poles, as
	\begin{equation}
		G(s) = \frac{M}{(s + \beta_{1}) (s + \beta_{2}) (s + \beta_{3})},
		\label{eq:bistable:limitcycle}
	\end{equation}
	and assume for simplicity that $0 < \beta_{1} < \beta_{2} < \beta_{3}$.
	By noticing that $G(s)$ has relative degree $n - q = 3$ and $r = 0$, it follows from Table 
	\ref{table:ppassivity:conditions} that $G(s)$ is only compatible with 
	$1$-passivity and $2$-passivity. For instance, simple computations yield,
	\begin{multline}
		\Re \{G(j \omega-\lambda)\} =
		\\
		M \frac{ \left( 3 \lambda - \beta_{1} - \beta_{2} - 
		\beta_{3} \right) \omega^2 + (\beta_{1} - \lambda)(\beta_{2} - \lambda)(\beta_{3} -
	\lambda)}{\left[ \omega^{2} + 
		( \beta_{1} - \lambda )^{2} \right] \left[ \omega^{2} + (\beta_{2} - \lambda)^{2}
		\right] \left[ \omega^{2} + (\beta_{3} - \lambda)^{2} \right] }.
		\label{eq:real:G_jw}
	\end{multline}
	Consider $M < 0$. By Theorem \ref{th:kypLemma},
	 a sufficient condition for $1$-passivity of $G(s)$ with rate $\lambda$ is
	\begin{equation}
		\beta_{1} < \lambda < \min \left\{ \beta_{2}, \frac{\beta_{1} + \beta_{2} + \beta_{3}}{3}
		\right\}.
		\label{eq:lambda:1-passive}
	\end{equation}
	Consider now $M > 0$. If
	\begin{equation}
		\max \left\{ \beta_{2}, \frac{\beta_{1} + \beta_{2} + \beta_{3}}{3} \right\}
		< \lambda < \beta_{3},
		\label{eq:lambda:2-passive}
	\end{equation}
	then $G(s)$ is $2$-passive. We will return to this example in the next section, to illustrate the
	design of a bistable system and of a system with a periodic attractor.
\end{example}
\section{Differential analysis of Lur'e feedback systems in the frequency domain}
\label{section:lure:frequency}
\subsection{Differential analysis of  Lur'e systems}
We will now apply the theory of the paper to a differential analysis of  Lur'e systems.
We look for conditions under which a Lur'e feedback system formed by a linear 
time-invariant system in feedback interconnection with a memoryless, time-independent nonlinearity
\begin{equation}
\label{eq:lureSys}
	\begin{cases}
		\dot{x} = A x + B u
		\\
		y = C x + D u
		\\
		u = - \varphi(y) \ .
	\end{cases}
\end{equation}
The prolonged system is given by \eqref{eq:lureSys} and 
\begin{equation}
\label{eq:diff_lureSys}
	\begin{cases}
		\dot{\delta x} = A \delta x + B \delta u
		\\
		\delta y = C \delta x + D \delta u
		\\
		\delta u = - \partial\varphi(y) \delta y \ .
	\end{cases}
\end{equation}
In the differential setting the usual sector condition on the nonlinearity
is replaced by a sector condition on its linearization.
Namely, the absolute $p$-dominance problem consists in finding conditions
under which  the linear time-invariant system in
feedback interconnection with memoryless, time-independent nonlinearity that satisfies
the differential sector condition
\begin{equation}
	\left( \partial\varphi(y)\delta y - K_{1} \delta y \right)^{\top} 
	\left( \partial\varphi(y)\delta y - K_{2} \delta y \right) \leq 0
	\label{eq:sector:condition}
\end{equation}
is $p$-dominant. The usual assumption is that $K_{1}$ and $K_{2}$ are matrices such that
$K_{2} - K_{1}$ is symmetric positive definite.  
We observe that in the scalar case the sector condition \eqref{eq:sector:condition}
is in fact equivalent to require that $\partial \varphi(y) \in [K_{1}, K_{2}]$ for all $y \in \RE$. 
Indeed, in the differential setting, the sector condition is a restriction on the slope of the nonlinearity.
Henceforth, we adopt the notation $\partial \varphi(y) \in [K_{1}, K_{2}]$ to denote the sector condition
\eqref{eq:sector:condition}.
\subsection{Kalman's conjecture and conditions for $p$-dominance of Lur'e systems}
An obvious {\it necessary} condition for $p$-dominance of  \eqref{eq:lureSys} with the differential
sector condition $\partial \varphi(y) \in [K_{1}, K_{2}]$ is $p$-dominance of the linear system
\begin{equation}
\label{eq:diff_lureSysbis}
	\begin{cases}
		\dot{\delta x} = A \delta x + B \delta u
		\\
		\delta y = C \delta x + D \delta u
		\\
		\delta u = - K \delta y \ .
	\end{cases}
\end{equation}
for any matrix $K$ satisfying $K \in [K_{1}, K_{2}]$. A famous conjecture formulated by Kalman in 1957 is
that this necessary condition is also sufficient.  
The following counterexample from \cite{bragin2011} shows that Kalman's 
conjecture fails in dimension $n = 4$.
\begin{example} \label{ex:kalman:counterexample}
	Consider a linear system of the form \eqref{eq:lureSys} where
	\begin{align}
		\label{eq:kalman:counter}
		\nonumber A & := 
		\begin{bmatrix}
			0 & -1 & 0 & 0
			\\
			1 & 0 & 0 & 0
			\\
			0 & 0 & 0 & 1
			\\
			0 & 0 & -1 & -1
		\end{bmatrix}, \;
		B := 
		\begin{bmatrix}
			10 \\ 10.1 \\ 0 \\ -1
		\end{bmatrix}, \\
		C & := 
		\begin{bmatrix}
			1 & 0 & -10.1 & -0.1
		\end{bmatrix}, \; D := 0.
	\end{align}
	Consider the static nonlinearity $\varphi(y) = \tanh(y) + \varepsilon y$, 
	where $\varepsilon = 0.01$. $\partial \varphi$ belongs to the
	sector $[\varepsilon, 1 + \varepsilon]$. Furthermore, the closed loop 
	of the linear system with $u = -K y$ is asymptotically stable for any 
	$K \in [0, +\infty)$. Nevertheless, simulations show that for the initial condition
	$x_{0} = [3, -3, -2, 0]^{\top}$ the system's trajectories display an oscillatory behavior,
	contradicting the stability claimed by the conjecture.
\end{example}
Kalman's conjecture is an early example of differential analysis. 
Writing down the conditions of the conjecture within an LMI formulation
shows that the condition of asymptotic stability of the linear systems with 
feedback gains $K \in [K_{1}, K_{2}]$ is equivalent to the existence of a family of symmetric 
positive definite matrices $P_{K}$ such that
\begin{equation}
	\left( A + B K C \right)^{\top} P_{K} 
	+ P_{K} \left( A + B K C \right) < 0,
	\label{eq:lmi:kalman}
\end{equation}
for $K \in [K_{1}, K_{2}]$. The difference between Kalman's 
conjecture and \eqref{eq:LMI-dominance:diff} for strict $0$-dominance is that  (\ref{eq:lmi:kalman}) does not
enforce a {\it constant} solution of the LMI. The dominance is only imposed {\it pointwise}, allowing for a different matrix $P_K$ for each
gain $K$. Indeed, Kalman conjecture provides a necessary condition for $0$-dominance. 
Searching for a constant solution $P$ in \eqref{eq:LMI-dominance:diff} enforces a stronger condition
than stability; it entails contraction.
The breakthrough in the history of absolute stability theory came from the circle criterion, which connects
a frequency domain property of the LTI system to a dissipativity property of the nonlinearity.
The theorem below provides the analog sufficient condition for the $p$-dominance
of Lur'e systems.
\begin{theorem} \label{th:quasiCircle}
	Consider a strictly $p$-dissipative linear system \eqref{eq:linear-open} 
	with rate $\lambda$ and supply \eqref{eq:supply}.
	Then, the feedback system \eqref{eq:lureSys} is strictly $p$-dominant with rate $\lambda$
	for all nonlinearities $\varphi$ that satisfy 
	\begin{equation}
		\begin{bmatrix}
			I\\ - \partial\varphi(y) 
		\end{bmatrix}^{\top}
		\begin{bmatrix}
			Q & L \\ L^{\top} & R
		\end{bmatrix}
		\begin{bmatrix}
			I \\ - \partial\varphi(y) 
		\end{bmatrix} \leq 0 \qquad \forall y \in \RE^{m_1} .
	\label{eq:generalised:sector:condition}
	\end{equation}
\end{theorem}
\begin{proofof}\emph{Theorem \ref{th:quasiCircle}.}
	From the strict dissipativity of the linear dynamics
	we have that the prolonged dynamics satisfies
\vspace{-8mm}{\small{$$
 		\begin{bmatrix}
			A \delta x \!+\! B \delta u \\ \delta x
		\end{bmatrix}^{\!\!\!\top} \!\!\!\!
		\begin{bmatrix}
			0 & P \\ P \!&\! 2\lambda P \!+\! \varepsilon I\!
		\end{bmatrix}\!\!\!
		\begin{bmatrix}
			A \delta x \!+\! B \delta u\\ \delta x
		\end{bmatrix}
		\!\leq\!
 		\begin{bmatrix}
			\!\delta y\! \\ \!\delta u\!
		\end{bmatrix}^{\!\!\!\top} \!\!\!\!
		\begin{bmatrix}
			\!Q\! & \!L\! \\ \!L^{\top}\! & \!R\!
		\end{bmatrix}\!\!\!
		\begin{bmatrix}
			\!\delta y\!\! \\ \!\delta u\!\!
		\end{bmatrix}\!.
$$}}
For $\delta u = -\partial \varphi(y) \delta y$,
using \eqref{eq:generalised:sector:condition} we get
$$
 		\begin{bmatrix}
			A \delta x \! - \! B \partial \varphi(y) \delta y \\ \delta x
		\end{bmatrix}^{\!\!\top} \!\!\!
		\begin{bmatrix}
			0 & P \\ P & 2\lambda P \!+\! \varepsilon I
		\end{bmatrix}\!\!
		\begin{bmatrix}
			A \delta x \! - \! B \partial \varphi(y) \delta y\\ \delta x
		\end{bmatrix}
		\!\!\leq \!0 
$$	
which shows that \eqref{eq:diff-p-dominance} holds along the trajectories
of the prolonged closed-loop system \eqref{eq:lureSys},\eqref{eq:diff_lureSys}.
\end{proofof}

Condition \eqref{eq:generalised:sector:condition} is in fact a generalized
sector condition for the linearization $\partial \varphi(y) \delta y$. 
In fact, for $Q = K_{1}^{\top} K_{2} + K_{2} K_{1}^{\top}$, 
$L = K_{1} + K_{2}$ and $R = 2 I$ we recover the differential sector 
condition \eqref{eq:sector:condition}.
We also observe that for $R = 0$, $Q = 0$ and $L = I$ the linear system
is $p$-passive and \eqref{eq:generalised:sector:condition} is equivalent to
the condition $\partial\varphi(y) \delta y \geq 0$.
Indeed, the slope of the nonlinearity belongs to the sector $[0, +\infty)$, that is,
its linearization is $0$-passive (the nonlinearity is monotone).
In this special case, Theorem \ref{th:quasiCircle} reduces to \cite[Proposition 9]{forni2017},
which provides a generalization of the passivity theorem to the context of $p$-dissipativity.
\subsection{The circle criterion for $p$-dominance}
Theorem \ref{th:quasiCircle} allows for a number of useful reformulations in the frequency domain.
\begin{corollary}\label{cor:quasiCircle}
	Let $G(s)$ be the transfer function of \eqref{eq:linear-open} and suppose that $\varphi$ satisfies
	\eqref{eq:generalised:sector:condition} with $Q = 0$. If $G(s)$ has $p$ poles in the interior
	of $\Omega_{\lambda}$, $n-p$ poles in the complement of $\Omega_{\lambda}$, 
	and the transfer function
	\begin{equation} 
		Z(s-\lambda) := L G(s-\lambda) + \frac{1}{2} R \label{eq:generalized:function}
	\end{equation}
	satisfies 
	\begin{equation} \label{eq:composite:passivity}
		Z(j \omega-\lambda) + Z^{*}(j \omega-\lambda) > 0 
	\end{equation}
	for all $\omega \in \RE \cup \{\infty\}$, then the feedback system \eqref{eq:lureSys} is strictly $p$-dominant with
	rate $\lambda$.
\end{corollary}
\begin{proofof}\emph{Corollary \ref{cor:quasiCircle}.}
	First note that $Z(s)$ is the transfer function of 
	\begin{equation} 
		 \begin{cases}
			\dot{x} = A x+ B u
			\\
			\bar{y} = L y + \frac{1}{2} R u.
		\end{cases} \label{eq:transformed-linear-system}
	\end{equation}
	Hence, the fact that $G(s)$ has $p$ poles in the interior of $\Omega_{\lambda}$ together with 
	\eqref{eq:composite:passivity} allow us to apply Theorem \ref{th:kypLemma}, which guarantees that the
	prolonged dynamics of \eqref{eq:transformed-linear-system} satisfies
	\begin{displaymath}
		\begin{bmatrix}
			\dot{\delta x} \\ \delta x
		\end{bmatrix}^{\top}
		\begin{bmatrix}
			0 & P 
			\\
			P & 2\lambda P + \varepsilon I
		\end{bmatrix}
		\begin{bmatrix}
			\dot{\delta x} \\ \delta x
		\end{bmatrix} 
		\leq
		\begin{bmatrix}
			\delta \bar{y} \\ \delta u
		\end{bmatrix}^{\top} 
		\begin{bmatrix}
			0 & I
			\\
			I & 0
		\end{bmatrix}
		\begin{bmatrix}
			\delta \bar{y} \\ \delta u
		\end{bmatrix}.
	\end{displaymath}
	Note also that 
	$$
	\mymatrix{c}{\delta \bar{y} \\ \delta u}^T \mymatrix{cc}{0 & I \\ I & 0} \mymatrix{c}{\delta \bar{y} \\ \delta u}
	= \mymatrix{c}{\delta y \\ \delta u}^T \mymatrix{cc}{0 & L \\ L^{\top}& R} \mymatrix{c}{\delta y \\ \delta u} \leq 0
	$$
	Thus, the result follows from Theorem \ref{th:quasiCircle}.
\end{proofof}

For $R = 2 I$, $Q = 0$ and $L = K = K^{\top} > 0$,
$\partial\varphi(y)$ belongs to the sector $[0, K]$ and we recover (and extend to the differential setting)
the standard formulation of the Positivity Theorem \cite[Theorem 5.18]{haddad2008} 
for $\lambda = 0$. It should also be noted that the nonlinearity $\varphi(y) = 0$ 
always satisfies the sector condition
\eqref{eq:generalised:sector:condition} if $Q = 0$. Therefore, for $Q=0$, 
a necessary condition for the $p$-dominance of the closed loop \eqref{eq:lureSys} 
is the $p-$dominance of linear system \eqref{eq:linear-open}.
Corollary \ref{cor:quasiCircle} can be further extended via 
loop transformations, leading to the following reformulation of the circle criterion.
\begin{corollary}
	[Multivariable Circle criterion] \label{cor:circle}
	Let $G(s)$ be the transfer function of \eqref{eq:linear-open} and suppose that 
	the nonlinearity $\varphi$ satisfies \eqref{eq:sector:condition} for
	some matrices $K_{1}, K_{2}$ such that $K_{2} - K_{1}$ is symmetric and positive definite.
	Then, the closed-loop system is strictly $p$-dominant if the transfer function
	\begin{equation}
		\tilde{G}(s) := G(s) (I + K_{1} G(s))^{-1}
		\label{eq:transfer:closed}
	\end{equation}
	has $p$-poles in the interior of $\Omega_{\lambda}$, $n-p$ poles in the complement
	of $\Omega_{\lambda}$, and
	\begin{equation}
		Z(s-\lambda) = \left( I + K_{2} G(s-\lambda) \right)
		\left( I + K_{1} G(s-\lambda) \right)^{-1},
		\label{eq:tranfer:loopTransform}
	\end{equation}
	satisfies \eqref{eq:composite:passivity}.
\end{corollary}
\begin{proofof}
	\emph{Corollary \ref{cor:circle}.}
	The proof follows directly from the loop transformation depicted in Figure \ref{fig:loopTransformation}
	and Corollary \ref{cor:quasiCircle}. Indeed, the loop transformation arises by considering 
	$\tilde{\varphi}(y) := \varphi(y) - K_{1} y$. Thus,
	\eqref{eq:sector:condition} transforms into
	\begin{equation}
		\label{eq:sector_special}
		\delta y^{\top} \partial \tilde{\varphi}(y)^{\top} \left( \partial \tilde{\varphi}(y) 
		\delta y - (K_{2} - K_{1}) \delta y \right) \leq 0.
	\end{equation}
	In other words, $\tilde{\varphi}$ satisfies
	\eqref{eq:generalised:sector:condition} with $R = 2I$, $Q = 0$, and
	$L = (K_{2} - K_{1})$. Furthermore, notice that the feedback of $G(s)$ and $\partial \varphi$ is
	equivalent to the one in Figure \ref{fig:loopTransformation}, where 
	$\tilde{G}(s)$ is as in \eqref{eq:transfer:closed}. Hence, $Z(s)$ in \eqref{eq:generalized:function}
	takes the form
	\begin{align*}
		Z(s) & = L \tilde{G}(s) + \frac{1}{2} R
		\\
		& = (K_{2} - K_{1})\tilde{G}(s) + I
		\\
		& = (I + K_{2}G(s))(I + K_{1} G(s))^{-1},
	\end{align*}
	and the result follows by Corollary \ref{cor:quasiCircle}.
	\begin{figure}[htpb]
		\centering
		\includegraphics[width=0.3\textwidth]{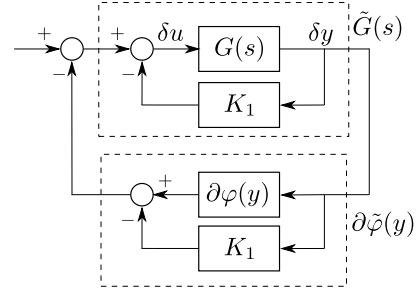}
		\caption{Loop transformation of the closed-loop system \eqref{eq:lureSys}.}
		\label{fig:loopTransformation}
	\end{figure}
\end{proofof}

Graphical conditions can be derived for SISO systems. In the next corollary
we will use $\Omega_{\lambda}$ to denote the Nyquist region
defined in \eqref{eq:nyquist_region}, and $\mathfrak{D}(K_1,K_2)$ to denote the 
disk in the complex plane given by the set 
$$
		\left\{ x \!+\! j y \!\in\! \mathbb{C} \,\bigg| \left( x \!+\! \frac{K_{1} + 
		K_{2}}{2 K_{1}K_{2}}\right)^{2} \!+\! y^{2}  \leq
		\left( \frac{K_{2} - K_{1}}{2 K_{1}K_{2}} \right)^{2}
		\right\}
$$
where $K_1$ and $K_2$ are real constants.
\begin{corollary} \label{corollary:nyquist:circle} 
	Consider the closed-loop system \eqref{eq:lureSys} given by 
	the linear system \eqref{eq:linear-open} with transfer function $G(s)$
	and by a static nonlinearity $\varphi$. \eqref{eq:lureSys} is strictly $p$-dominant if
	\begin{enumerate}
	\item[i)] the nonlinearity $\varphi$ satisfies  \eqref{eq:sector:condition};
	\item[ii)] $G(s)$ has no poles on the boundary of $\Omega_{\lambda}$;
	\item[iii)] the Nyquist plot of $G(s)$ along the boundary of  
		$\Omega_{\lambda}$ makes
		$E = p - q$ encirclements of
		the point $-1/K_{1}$ in the clockwise direction, where $q$ is the number of 
		poles of $G(s)$ in $\Omega_{\lambda}$; 
	\item[iv)] one of the following conditions is satisfied
			\begin{enumerate}
				\item $0 < K_{1} < K_{2}$ and the Nyquist plot of
				$G(s)$ along the boundary of $\Omega_{\lambda}$
				lies outside of the disc $\mathfrak{D}(K_{1}, K_{2})$.
				\item $K_{1} < 0 < K_{2}$ and the Nyquist plot of
					$G(s)$ along the boundary of $\Omega_{\lambda}$
					lies inside the disc $\mathfrak{D}(K_{1}, K_{2})$.
				\item $K_{1} < K_{2} < 0$ and the Nyquist plot of
					$G(s)$ along the boundary of $\Omega_{\lambda}$
					lies outside the disc $\mathfrak{D}(K_{1}, K_{2})$.
			\end{enumerate}
	\end{enumerate}
\end{corollary}
\begin{proof}
	The proof relies on the loop transformation of Figure \ref{fig:loopTransformation} and
	Corollary \ref{cor:quasiCircle}. Indeed, notice that
	$i)$ guarantees that $\tilde{\varphi}(y) = \varphi(y) - K_{1} y$ satisfies
	\eqref{eq:generalised:sector:condition}
	for $L = K_{2} - K_{1}$, $R = 2 I$ and $Q = 0$.  Now, define $Z(s)$ as in 
	\eqref{eq:generalized:function} but applied to the feedback interconnection of $\tilde{G}(s)$ and
	$\partial \tilde{\varphi}(y)$, that is $Z(s) = L \tilde{G}(s) + \frac{1}{2} R$, where
	\begin{displaymath}
		\tilde{G}(s) = \frac{G(s)}{1 + K_{1}G(s)}.
	\end{displaymath}
	The well-posedness of $\tilde{G}(j \omega - \lambda)$ follows by $ii)$. Next, by $iii)$ and Theorem 
	\ref{theorem:nyquistCriterion} it follows that $\tilde{G}(s)$ is strictly $p$-dominant. Thus, it 
	remains to prove that $Z(j \omega - \lambda)$ satisfies \eqref{eq:composite:passivity} in order to
	guarantee that all the assumptions of Corollary \ref{cor:quasiCircle} are satisfied and the result
	follows. 
	In the scalar case \eqref{eq:composite:passivity} is equivalent to
	\begin{equation} \label{eq:circle:condition}
		\Re \left\{ \frac{1 + K_{2} G(j \omega-\lambda)}{1 + K_{1} G(j \omega - \lambda)}
	\right\} > 0, \text{ for all } \omega \in \RE \cup \{\infty\}.
	\end{equation}
	
	Let $X(\omega)$ and $Y(\omega)$ be respectively the
	real and imaginary parts of $G(j \omega-\lambda)$, that is, $G(j \omega-\lambda)
	= X(\omega) + j Y(\omega)$. Straightforward computations reveal that \eqref{eq:circle:condition} 
	is equivalent to
	\begin{equation}
		K_{1} K_{2} Y(\omega)^{2} + \left( K_{1} X(\omega) + 1 \right)
		\left( K_{2} X(\omega) + 1 \right) > 0, 
		\label{eq:circle:condition:2}
	\end{equation}
	for all $\omega \in \RE \cup \{\infty\}$. Now, assuming that $0 < K_{1} < K_{2}$,
	\eqref{eq:circle:condition:2} can be rewritten as
	\begin{displaymath}
		\left( X(\omega) + \frac{K_{1} + K_{2}}
		{2 K_{1} K_{2}} \right)^{2} + Y(\omega)^{2} > \frac{\left( K_{2} - K_{1} \right)^{2}}
		{4 K_{1}^{2} K_{2}^{2}}
	\end{displaymath}
	which requires that the Nyquist plot of $G(s)$ along the boundary of 
	$\Omega_{\lambda}$ must lie 
	outside the disc $\mathfrak{D}(K_{1}, K_{2})$, leading to  $iv.a)$. The other two cases are similar.
\end{proof}
The reader will recognize that the above corollaries extend the classical
circle criterion to the analysis of attractors that are not necessarily fixed points.
The next example shows how to use these tools to give 
insights on the existence of oscillatory behaviors in Lur'e systems.
\begin{example} \label{example:positive:negative}
	We revisit Example \ref{ex:3-poles} with $M > 0$,
	for $\varphi(y) = a \tanh(k y)$, where $a$, $k$ are
	positive real numbers. Notice that $\varphi(y)$ satisfies \eqref{eq:sector:condition}
	with $K_{1} = 0$ and $K_{2} = ak$.
	
	\underline{Negative feedback}: the closed loop is $2$-passive. In fact, 
	looking at the second case in Example \ref{ex:3-poles},
	the transfer function $G(s)$ of the linear part is $2$-passive with rate $\lambda$ satisfying 
	\eqref{eq:lambda:2-passive}. 
	A direct application of the Corollary \ref{corollary:nyquist:circle} 
	reveals that the closed loop is $2$-dominant whenever the Nyquist plot of $G(s)$ along 
	the boundary of $\Omega_{\lambda}$ satisfies $iii)$ and $iv.a)$ for any
	nonlinearity in the differential sector $\sector[K_{1}, K_{2}]$. 
	For example, setting $\lambda = 2.6$, and the system parameters 
	$M = 1$, $\beta_{1} = 1$, $\beta_{2} = 2$, $\beta_{3} = 3$, $a = 10$ and $k = 10$, it becomes clear
	that $ii)$ in Corollary \ref{corollary:nyquist:circle} holds. Furthermore, with the selected values of
	parameters, it follows that $G(s)$ has $2$ poles in $\Omega_{\lambda}$, hence $iii)$ in 
	Corollary \ref{corollary:nyquist:circle} asks for $E = 2 - 2 = 0 $ encirclements of the 
	point $-1/K_{1}$.
	Figure \ref{fig:ex:limitCycle} reveals that the Nyquist plot of $G(s)$ along the boundary of 
	$\Omega_{\lambda}$ lies in $\mathbb{C}_{+}$, that is, $iii)$ and $iv.a)$ also hold for any
	$k \in (0, +\infty)$. Hence, the $2$-dominance of 
	the closed-loop with rate $\lambda = 2.6$ follows.
	\underline{Positive feedback}: in this case, the addition of the constant multiplier $-1$
	leads us to consider the negative feedback case as above but with $M < 0$.
	Thus, from the
	first part of Example \ref{ex:3-poles} it follows that the $G(s)$ is $1$-passive with rate
	$\lambda$ satisfying \eqref{eq:lambda:1-passive}. 
	From Corollary \ref{corollary:nyquist:circle}, the closed loop is $1$-dominant whenever the Nyquist plot of $G(s)$ along 
	the boundary of $\Omega_{\lambda}$ satisfies $iii)$ and $iv.a)$. 
	For example, setting the rate $\lambda = 2.6$ and the system parameters 
	$M = -10$, $\beta_{1} = 2$, $\beta_{2} = 3$, $\beta_{3} = 5$, $a = 10$ and $k = 10$,
	it follows that the open-loop
	has $1$ pole in the interior of $\Omega_{\lambda}$ and because the Nyquist plot of 
	$G(j \omega - \lambda)$ falls in $\mathbb{C}_{+}$ (see Figure
	\ref{fig:ex:bistable}), the conditions $ii), iii)$ and $iv.a)$ hold for any
	$k \in (0 +\infty)$, which proves  $1$-dominance of the closed-loop.
	\begin{figure}[htpb]
		\centering
		\includegraphics[width=0.3\textwidth]{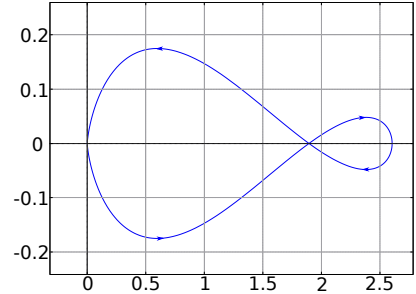}
		\caption{Nyquist plot of the system \eqref{eq:bistable:limitcycle} (with parameters
		$\beta_{1} = 1$, $\beta_{2} = 2$, $\beta_{3} = 3$ and $M = 1$), along the 
		$\lambda$-shifted Nyquist path $\Omega_{\lambda}|_{\lambda = 2.6}$. Notice that in
		this case the system is $2$-passive.
		}
		\label{fig:ex:limitCycle}
	\end{figure}
	
	\begin{figure}[htpb]
		\centering
		\includegraphics[width=0.3\textwidth]{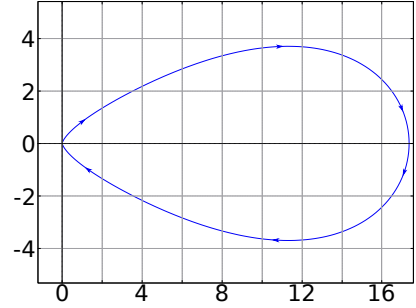}
		\caption{Nyquist plot of the system \eqref{eq:bistable:limitcycle} (with parameters
		$\beta_{1} = 2$, $\beta_{2} = 3$, $\beta_{3} = 5$ and $M = -10$), along the 
		$\lambda$-shifted Nyquist path $\Omega_{\lambda}|_{\lambda = 2.6}$. Notice that in this
	case the system is $1$-passive.}
		\label{fig:ex:bistable}
	\end{figure}
\end{example}
\section{The asymptotic behavior of dominant Lur'e systems}
\label{section:lure:differential}
In the previous sections we have extended a number of classical results to the analysis of  $p$-dominance.
We will now illustrate how this analysis can be used to analyze  the asymptotic behavior of
Lur'e systems.
\subsection{Contraction analysis ($0$-dominance)}
As a first example we briefly revisit the property of 
global contraction of the vector field, largely studied in the literature
(see e.g., \cite{lohmiller1998, pavlov2005, Fromion2005,jouffroy2010} and references therein).
Strict $0$-dominance implies contraction.
Indeed, Theorem \ref{th:quasiCircle} and its corollaries provide conditions
for global contraction. The zero equilibrium of a Lure feedback system is then necessarily
globally asymptotically stable.
\subsection{Bistability ($1$-dominance)}
From Corollary \ref{thm:asymptotic_behavior},  strict $1$-dominance is a useful tool for the analysis
of bistability. Global bistability of the Lure feedback system  \eqref{eq:lureSys} is ensured from the following three properties:
\begin{enumerate}
\item strict $1$-dominance of the closed loop; 
\item boundedness of solutions; and
\item the algebraic equation $ u + \varphi(G(0)u) = 0$ has  three isolated solutions.
\end{enumerate}
As an illustration, consider Example \ref{example:positive:negative}. 
Recall that, with positive feedback, any nonlinearity in the differential sector $(0, +\infty)$
guarantees strict $1$-dominance of the closed loop \eqref{eq:lureSys},
with rate $\lambda = 2.6$. This proves the $1$-dominance.
Trajectories of the closed-loop system are bounded because the input $u = -\varphi(y)$ is
by definition bounded and the linear system is BIBO stable.  
Finally, a graphical argument shows that the algebraic equation has three isolated solutions.
In conclusion, we have shown that the closed-loop system formed by \eqref{eq:bistable:limitcycle} in 
positive feedback interconnection with $\varphi(y) := a \tanh(k y)$ is globally bistable, that is, every
solution converges to one of the three fixed points, out of which one is unstable. Figure \ref{fig:lure:bistable} confirms
the predicted behavior.
\begin{figure}[htpb]
	\centering
	\includegraphics[width=0.35\textwidth]{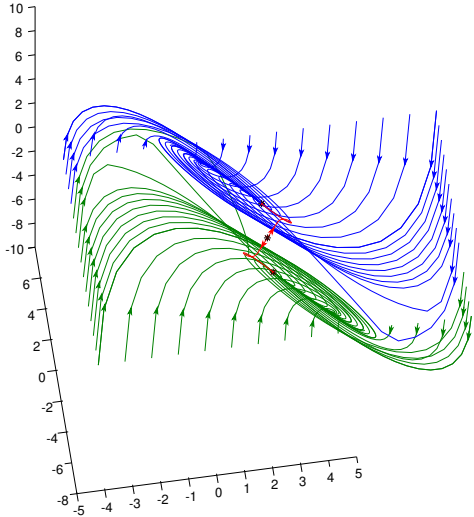}
	\caption{Phase trajectories of the system \eqref{eq:bistable:limitcycle} in a positive
	feedback interconnection with the nonlinearity $\varphi(y) = a \tanh(k y)$ with parameters 
	$M = 10$, $\beta_{1} = 2$, $\beta_{2} = 3$, $\beta_{3} = 5$, $a = 10$ and $k = 10$. 
	The red trajectory represents the heteroclinic solution.}
	\label{fig:lure:bistable}
\end{figure}
The proposed analysis is also useful for control design. 
If the open-loop system does not have the desired 
degree of dominance, a controller can be introduced to shape the 
frequency response in such a way that the assumptions of Corollary \ref{corollary:nyquist:circle} are met.
For illustration, we study the bistability of the Lur'e system arising from 
the interconnection of 
\begin{equation}
	G(s) = \frac{3 (s + 1)}{(s^{2} + 4 s + 8)(s+3)} \ .
	\label{eq:tf:1mono}
\end{equation}
with a  saturating input-output characteristic in the differential sector $[1, 5]$. 
We consider the case of
positive feedback. Note that the open-loop system $G(s)$ is not $1$-dominant for any value of 
rate $\lambda$, since it has dominant complex conjugated poles at $s = -2 \pm j 2$. 
To achieve strict $1$-dominance of the closed loop we enforce
strict $1$-dominance of the return ratio by pairing $G(s)$ with 
a controller $C(s)$. Set the desired rate to 
$\lambda = 2.1$ and observe that $G(s-\lambda)$ has two dominant poles at 
$s = 0.1 \pm j 2$ and one pole at $s = -0.9$.
Take $C(s-\lambda) := 0.4/(s - 1.9)$.
The Nyquist plot of $-KG(s-\lambda)C(s-\lambda)$  
is depicted in Figure \ref{fig:nyquist:1mono}. 
The main idea behind the selection of the controller 
$C(s-\lambda)$ is to increase the phase change in the Nyquist plot.
By adding an unstable pole in the shifted-system, the
associated Nyquist plot reflects a change in phase of 180 degrees, at zero frequency.
Thus, we can set the value of the gain $K$ to achieve
$2$ encirclements in the counterclockwise direction around the disk $\mathfrak{D}(1,5)$.
$1$-dominance follows.
It is noteworthy that the transfer function of the desired controller reads $C(s) = 0.4 / (s + 0.2)$, 
which is a simple first-order lag. Now, setting $\varphi(y) := y + \tanh(4y)$,
the same analysis as above shows that the system has three equilibria, that all solutions are
bounded and the origin is unstable. Hence, the system is bistable.
\begin{figure}[htpb]
	\centering
	\includegraphics[width=0.4\textwidth]{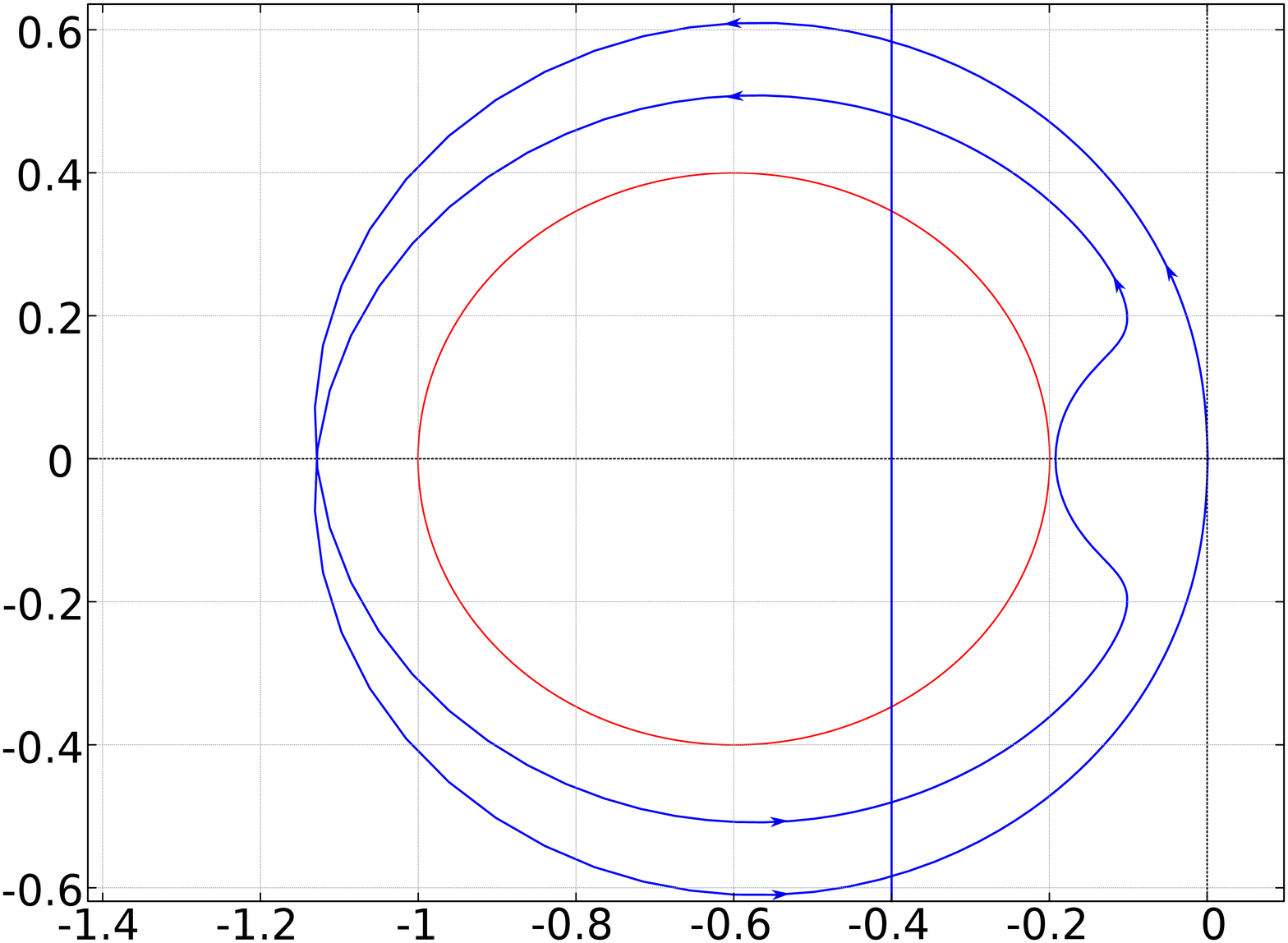}
	\caption{Nyquist plot of $-KG(s-\lambda)C(s-\lambda)$ for $s=j\omega$, $\omega \in \real$,
	for $\lambda = 2.1$.}
	\label{fig:nyquist:1mono}
\end{figure}
\subsection{Limit cycle oscillation ($2$-dominance)}
To prove the convergence of solutions to a limit cycle in a Lure feedback system, we verify
the following three conditions: 
\begin{enumerate}
\item strict $2$-dominance of the closed loop; 
\item boundedness of the solutions;
\item a forward invariant region that does not contain fixed points.
\end{enumerate}
By Corollary \ref{thm:asymptotic_behavior}, those three conditions imply that all trajectories 
with initial condition in the forward invariant region converge to a periodic attractor.
Returning to Example \ref{example:positive:negative}, recall that 
the negative feedback of $G(s)$ with any
nonlinearity in the differential sector $(0,+ \infty)$ 
gives a strictly $2$-dominant closed loop. 
We conclude that all trajectories that do not converge to the unstable fixed point
necessarily converge to a limit cycle. This is illustrated Figure \ref{fig:lure:limitCycle}.
\begin{figure}[htpb]
	\centering
	\includegraphics[width=0.35\textwidth]{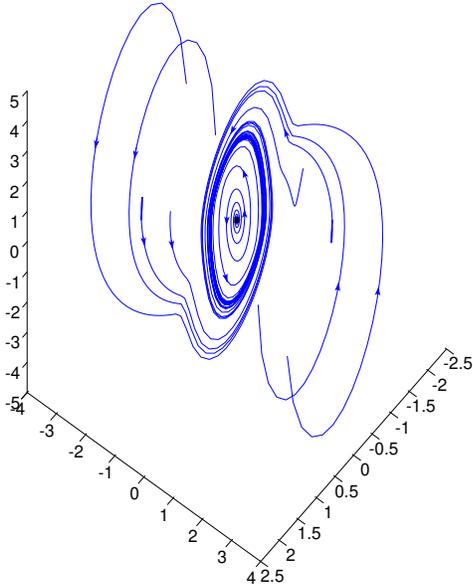}
	\caption{Phase portrait of the system \eqref{eq:bistable:limitcycle} in feedback interconnection
	with the nonlinearity $\varphi(y) = a \tanh(k y)$ with parameters $M = 1$, $\beta_{1} = 1$,
	$\beta_{2} = 2$, $\beta_{3} = 3$, $a = 10$ and $k = 10$.}
	\label{fig:lure:limitCycle}
\end{figure}
\begin{remark}
The analysis of $2$-dominance says nothing about  the uniqueness of the limit cycle.
As an illustration, we will verify that the system  analyzed  in 
Example \ref{ex:kalman:counterexample} is  $2$-dominant, despite the existence of 'hidden' oscillations.
The nonlinearity $\varphi(y) = \tanh(y)  + \varepsilon y$
satisfies the differential sector condition $\frac{\partial \varphi}{\partial y} \in [K_{1},
	K_{2}]$, with $K_{1} = \varepsilon$, $K_{2} = 1 + \varepsilon$, where $\varepsilon = 0.0125$.
	By Corollary \ref{corollary:nyquist:circle}, the closed-loop system is
	$2$-dominant with rate $\lambda = 0.275$ for all the nonlinearities in the prior sector, see
	Figure \ref{fig:kalman:circle}. 
	\begin{figure}[htpb]
		\centering
		\includegraphics[width=0.4\textwidth]{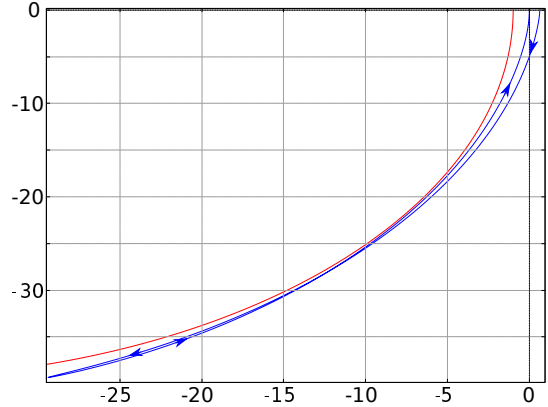}
		\caption{Nyquist plot of the system \eqref{eq:kalman:counter} for $\omega$
		running from $[0, +\infty)$, showing that the Nyquist plot falls outside of the 
			disk $\mathfrak{D}(K_{1}, K_{2})$ (in red). Thus,
		the system is $2$-dominant with rate $\lambda = 0.275$ for nonlinearities in
		the differential sector	$\sector[K_{1}, K_{2}]$.}
		\label{fig:kalman:circle}
	\end{figure}
\end{remark}
\subsection{Chaotic behavior ($3$-dominance)}
Chaotic behaviors may arise for larger dominance degree. 
Consider the Chua's circuit described by 
 \begin{displaymath}
	 \begin{cases}
		 \dot{x}_{1} = \alpha (x_{2} -x_{1} - u)
		 \\
		 \dot{x}_{2} = x_{1} - x_{2} + x_{3}
		 \\
		 \dot{x}_{3} = - \beta x_{2} \ .
	 \end{cases}
 \end{displaymath}
It is well known that this Lur'e system shows a chaotic behavior for certain range of 
parameters and nonlinearity $u = - \varphi(x_{1})$, \cite{matsumoto1985}.
In fact, with the parameters $\alpha = 8.8$, $\beta = 15$  and the nonlinearity $\varphi(x_{1}) = \tanh(2 x_{1}) + 0.7
x_{1}$ the double scroll attractor appears. 
An application of Corollary \ref{corollary:nyquist:circle} shows that the interconnection is
in fact $3$-dominant
with rate $\lambda = 4$. Indeed, considering $\lambda = 4$, the system has 
two complex eigenvalues in $\Omega_\lambda$. By Corollary \ref{corollary:nyquist:circle}
it follows that the system is $3$-dominant if the Nyquist plot encloses the disk $\mathfrak{D}(0.7, 2)$ once
in the clockwise direction, as shown in Figure \ref{fig:chuaNyquist}. Furthermore, the system is $3$-passive
with rate $\lambda > 9.67$.
\begin{figure}[htpb]
	\centering
	\includegraphics[width=0.4\textwidth]{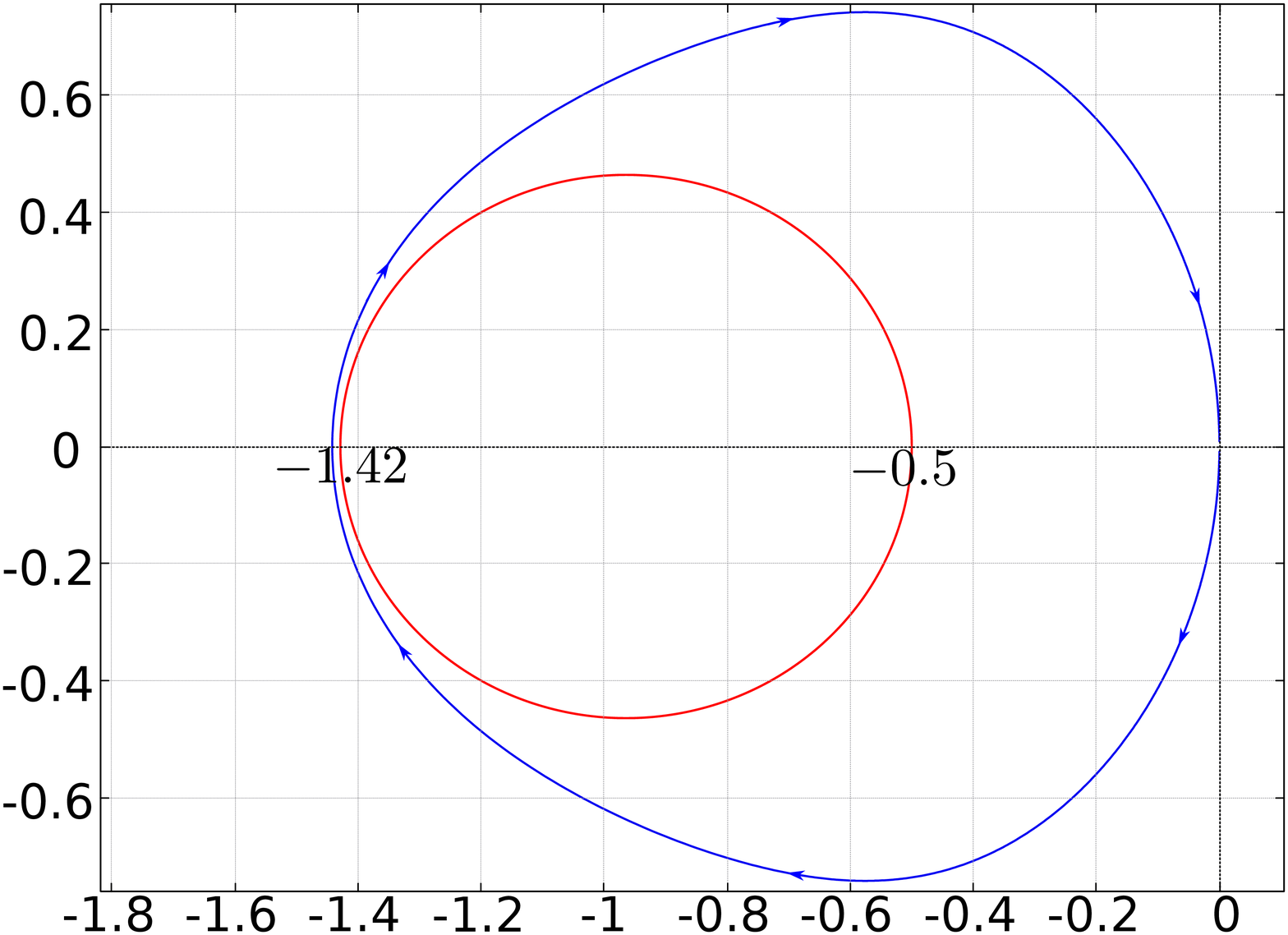}
	\caption{Nyquist diagram of the linear part of Chua's circuit with rate $\lambda = 4$.}
	\label{fig:chuaNyquist}
\end{figure}
\section{Conclusions}
Dominance analysis provides sufficient conditions for the asymptotic behavior
of a nonlinear system to be low-dimensional. For LTI systems, the property
is verified by solving a linear matrix inequality and checking the inertia of the 
solution. The KYP lemma provides a frequency-domain characterization of this
property. We have illustrated the potential of the frequency domain characterization
in the analysis of Lure's feedback systems. The analysis provides graphical tests
for $p$-dominance such as the circle criterion. More fundamentally, it provides
robustness margins very much like in stability analysis. The theory has been
illustrated on the analysis of bistable $1$-dominant systems and oscillatory
$2$-dominant systems.

\end{document}